\documentclass[journal,10pt,onecolumn,draftclsnofoot]{IEEEtran}
\usepackage{graphicx}
\usepackage{amsmath}
\usepackage{latexsym}
\usepackage{graphicx}
\usepackage{bm}
\usepackage{amssymb}
\usepackage{array}
\usepackage{setspace}
\usepackage{fancyhdr}
\usepackage{amsmath, amssymb, epsfig}
\usepackage{float}
\usepackage{subcaption}

\usepackage{color}

\usepackage{csquotes}
\MakeOuterQuote{"}

\makeatletter
\newcommand{\vast}{\bBigg@{3.2}}
\newcommand{\Vast}{\bBigg@{4.2}}
\makeatother

\newtheorem{theorem}{Theorem}

\begin{document}
	%
	\title{\linespread{1} Performance Analysis and Optimization of RIS-Assisted Networks in Nakagami-$m$ Environment}
	
	
	
	\author{\small \authorblockN{Monjed H. Samuh\authorrefmark{1},
			Anas~M.~Salhab\authorrefmark{2},~\IEEEmembership{\small Senior~Member,~IEEE},
			Ahmed H. Abd El-Malek\authorrefmark{3},~\IEEEmembership{\small Member,~IEEE}}\\
		\vspace{0.2cm}
		\authorrefmark{1}Applied Mathematics \& Physics Department,
		Palestine Polytechnic University, Hebron, Palestine,  e-mail: monjedsamuh@ppu.edu\\
		\vspace{0.2cm}
		\authorblockA{\authorrefmark{2}Electrical Engineering Department, King Fahd University of Petroleum \& Minerals, Dhahran 31261, Saudi Arabia\\
			e-mail: salhab@kfupm.edu.sa}\\
		\vspace{0.2cm}
		\authorrefmark{3}School of Electronics, Communications \& Computer Engineering, Egypt-Japan University of Science \& Technology, Alexandria 21934, Egypt, e-mail: ahmed.abdelmalek@ejust.edu.eg}

	
	
	\maketitle{}

	\begin{abstract}
		This work studies and optimizes the performance of reconfigurable intelligent surface (RIS)-aided networks in Nakagami-$m$ fading environment. First, accurate closed-form approximations for the channel distributions are derived. Then, closed-form formulas for the system outage probability, average symbol error probability (ASEP), and the channel capacity are obtained. Furthermore, we provide three different optimization approaches for finding the optimum number of reflecting elements to achieve a target outage probability. At the high signal-to-noise ratio (SNR) regime, a closed-form expression for the asymptotic outage probability is obtained to give more insights into the system performance. Results show that the considered system can achieve a diversity order of $\frac{a+1}{2}$, where $a$ is a function of the Nakagami-$m$ fading parameter $m$ and the number of reflecting elements $N$. Moreover, findings show that $m$ is more influential on the diversity order than $N$. Finally, the achieved expressions are applicable to non-integer values of $m$ and any number of meta-surface elements $N$.
	\end{abstract}


	%
	\IEEEpeerreviewmaketitle

\section{Introduction}
Recently, the topic of reconfigurable intelligent surfaces (RISs) gains a lot of attention as an efficient technique for future wireless communication networks. Being an artificial surface made of electromagnetic stuff, RIS can customize the propagation of the radio waves impinging upon it aiming to enhance the signal reception at destination \cite{Renzo1,Alouini1}. 
It enhances the spectrum and energy efficiency of wireless communication with a reasonable cost and complexity. 
	
The RIS technique's main characteristics and its possible applications have been discussed in \cite{Liang}. Basar \textit{et. al} provided in \cite{Basar1} a comprehensive discussion on the key differences of RIS with other techniques, state-of-the-art solutions of RIS, and the critical open research problems in this area of research. The RIS behavior was proven to outperform that of multi-antenna amplify-and-forward (AF) relaying and conventional massive multiple-input multiple-output networks in \cite{Wu}. Recently, \cite{Yang1} investigated the behavior of RIS-assisted hybrid indoor visible light communication-radio frequency (RF) network where both the bit error rate (BER) and outage probability of decode-and-forward (DF) and AF relaying schemes we provided. The secrecy performance of the RIS-aided system was investigated in the presence of eavesdropper and direct path in \cite{Yang5}.
	
Among the available literature on RIS-assisted networks, some works adopted the central limit theorem (CLT) in their derivations. These derivations have been noticed to give reasonable results only for high numbers of meta-surface elements \cite{Yang6}. And even at a high number of reflecting elements, still, the CLT results diverge from the actual behavior of RIS as we go further in increasing the average signal-to-noise ratio (SNR), as will be seen in our results. To solve this problem, the authors in \cite{Yang6} and \cite{Boulogeorgos} have obtained new channel distributions for RIS-aided networks assuming Raleigh fading channels and over Rician fading channels in \cite{AnasWCLs}. The error probability behavior of RIS-assisted network has been considered recently in \cite{Ferreira} under Nakagami-$m$ fading environment. The results were valid only for BPSK and $M$-QAM modulation techniques. Although the authors in \cite{Ferreira} considered Nakagami-$m$ fading channels, only exact expressions for the error probability were derived for a limited number of reflecting elements. In addition, they derived their approximate expressions and bounds for two specific modulation schemes, as mentioned before. Moreover, no insights into the system behavior at high SNR values were provided in that study. 

Motivated by aforementioned work, this work derives the channel distributions of RIS-aided systems adopting Nakagami-$m$ fading model using Laguerre series approach \cite{Primak}. These derived statistics of SNR are then utilized for achieving accurate approximations for the outage probability, average symbol error probability (ASEP), and the channel capacity. The obtained results are applicable to any number of meta-surface elements $N$ and non-integer values of fading parameter $m$. Furthermore, we utilize the outage probability results in formulating an optimization problem where the optimum number of reflecting elements is determined using three various optimization schemes. Moreover, simple expressions are derived for the asymptotic outage probability, diversity, and coding gains. Compared to \cite{AnasWCLs}, which studied the performance of RIS-aided networks over Rician fading channels, we provide here a new optimization problem for finding the optimum number of reflecting elements using two different approaches, in addition to the exact solution. Compared to Rician fading, which fits the case when there is one line-of-sight path between transmitter and receiver plus a number of other paths, Nakagami-$m$ fading is more versatile and can be used to model a large variety of fading channels by changing the value of its shape parameter $m$ in its equation \cite{Alouinib}. The provided results of diversity order and coding gain for RIS-aided networks over Nakagami-$m$ fading channels are new and being presented for the first time in this work.
	
The notation used in this paper is defined as follows: The term $f_{\beta}(x)$ is the probability distribution function (PDF) of random variables (RV) $\beta$, and $F_{\beta}(x)$ is its cumulative distribution function (CDF). The function $\gamma(\cdot,\cdot)$ denotes the lower incomplete gamma function \cite[Eq. (8.350.1)]{Grad.}.  Also, $K_{\nu}(\cdot)$ is the $\nu$-order second type modified Bessel function \cite[Eq. (8.432)]{Grad.}.  $pF_q\left(a_{1},...,a_{p};b_{1},...,b_{q};x\right)$ is the generalized hyper-geometric function \cite[Eq. (9.14.1)]{Grad.}.
	
	

	
	
	
	
	
	\section{System and Channel Models}\label{SCMs}
	
	For RIS-aided systems in \cite{Alouini1}, the maximized end-to-end (e2e) SNR can be expressed as 
	\begin{equation} \label{eq:gamma}
		\gamma =\frac {\left ({\sum _{i=1}^{N} \alpha _{i} \beta _{i} }\right)^{2} E_{s} }{N_{0}}= \bar{\gamma}Z^{2},
	\end{equation}
	where $\alpha_{i}\,(i=1,2,\ldots,N)$ are the first hop channel envelops and assumed to be independent and identically distributed (i.i.d.) Nakagami-$m$ random variables (RVs) each of shape parameter $m_1$ and scale parameter $\Omega_1$, and $\beta_{i}\,(i=1,2,\ldots,N)$ are the second hop channel envelops and also assumed to be i.i.d. Nakagami-$m$ RVs each with shape parameter $m_2$ and scale parameter $\Omega_2$. Moreover, $\alpha_{i}$ and $\beta_{i}$ are independent. $E_{s}$ is the average power of the transmitted signal, $N_{0}$ is the power of zero mean additive white Gaussian noise term, $N$ is the number of reflecting elements, and $\bar{\gamma}$ is the average SNR. For more details on the system and channel models, one can refer to \cite{Alouini1}. In the following, some statistical characterizations of $\gamma$ are presented.
	
	\begin{theorem} Assuming Nakagami-$m$ fading model, the probability density function (PDF) and the cumulative distribution function (CDF) of $\gamma$ can be, respectively, given by
		\begin{equation}\label{eq:PDFgamma}
			f_{\gamma}(\gamma)\simeq\frac{\left(\frac{\gamma }{\bar{\gamma}}\right)^{\frac{a}{2}} \exp \left(-\frac{\sqrt{\frac{\gamma }{\bar{\gamma}}}}{b}\right)}{2 b^{a+1} \Gamma (a+1) \sqrt{\bar{\gamma} \gamma }},
		\end{equation}
		and
		\begin{equation}\label{eq:CDFgamma}
			F_{\gamma}(\gamma)\simeq \frac{\gamma \left(a+1,\frac{\sqrt{\frac{\gamma }{\bar{\gamma}}}}{b}\right)}{\Gamma (a+1)},
		\end{equation}
		where
		
				$a=\frac{m_1 m_2 N \Gamma (m_1)^2 \Gamma (m_2)^2}{m_1 m_2 \Gamma (m_1)^2 \Gamma (m_2)^2-\Gamma \left(m_1+\frac{1}{2}\right)^2 \Gamma \left(m_2+\frac{1}{2}\right)^2}-N-1,$ and
		$b=\frac{m_1 m_2 \Gamma (m_1)^2 \Gamma (m_2)^2-\Gamma \left(m_1+\frac{1}{2}\right)^2 \Gamma \left(m_2+\frac{1}{2}\right)^2}{\sqrt{\frac{m_1}{\Omega_1}} \Gamma (m_1) \Gamma \left(m_1+\frac{1}{2}\right) \sqrt{\frac{m_2}{\Omega_2}} \Gamma (m_2) \Gamma \left(m_2+\frac{1}{2}\right)}.$
	\end{theorem}
	
	\begin{proof}
		Let $Z_i=\alpha _{i} \beta _{i} $ be the product of two Nakagami-$m$ RVs. The PDF of $Z_i$ is given by
		\begin{equation}\label{eq:PDFdoubleNaka}
f_{Z_i}(z)=\frac{4(\frac{m_1m_2}{\Omega_1\Omega_2})^{2(m_1+m_2)}}{\Gamma(m_1)\Gamma(m_2)} v^{m_1+m_2-1} K_{m_1-m_2}\left(2\sqrt{\frac{m_1}{\Omega_1}}\sqrt{\frac{m_2}{\Omega_2}} z\right).
\end{equation}

		Moreover, the mean and variance of $Z_{i}$ are, respectively given by
		\begin{equation}
			E(Z_i)=\frac{\Gamma \left(m_1+\frac{1}{2}\right) \Gamma \left(m_2+\frac{1}{2}\right)}{\sqrt{\frac{m_1}{\Omega_1}} \Gamma (m_1) \sqrt{\frac{m_2}{\Omega_2}} \Gamma (m_2)},
		\end{equation}
		and
		\begin{equation}
			Var(Z_i)=\Omega_1 \Omega_2 \left(1-\frac{\Gamma \left(m_1+\frac{1}{2}\right)^2 \Gamma \left(m_2+\frac{1}{2}\right)^2}{\Gamma (m_1)^2 \Gamma (m_2)^2}\right).
		\end{equation}
		Now, according to \cite[Sec. 2.2.2]{Primak}, the PDF of $Z=\sum _{i=1}^{N} Z_{i}$ can be tightly approximated by the first term of a Laguerre expansion. The values of $a$ and $b$ are related to the mean and variance of $Z$ as
		\begin{equation}\label{eq:paramA}
			a=\frac{\left(E(Z)\right)^2}{Var(Z)}-1,
		\end{equation}
		and
		\begin{equation}\label{eq:paramB}
			b=\frac{Var(Z)}{E(Z)},
		\end{equation}
		where $E(Z)=N E(Z_i)$ and $Var(Z)=N Var(Z_i)$. Therefore, we formulate the PDF of $Z$ such as
		\begin{equation}\label{eq:PDFofZ}
			f_{Z}(z)=\frac{z^{a}}{b^{a+1} \Gamma (a+1)} \exp \left(-\frac{z}{b}\right).
		\end{equation}
		Upon performing RV transformation, the PDF in \eqref{eq:PDFgamma} is achieved. Accordingly, the CDF is found using
		\begin{equation}\label{eq:CDFdef}
			F_{\gamma}(\gamma)=\int_{0}^{\gamma}f_{\gamma}(x) dx.
		\end{equation}
		The integral in \eqref{eq:CDFdef} is evaluated by substitution and then by employing \cite[Eq. (8.350.1)]{Grad.}. This concludes the proof.
	\end{proof}
	
	\section{Performance Analysis}\label{EPA}
	
	\subsection{Outage Probability}
	The system outage probability is defined as the event where the e2e SNR goes below a predetermined outage threshold $\gamma_{\sf out}$. Mathematically speaking $P_{\sf out}=\mathrm{Pr}\left[\gamma\leq\gamma_{\sf out}\right]$,
	where $\mathrm{Pr}[.]$ is the probability operation. Thus, using (\ref{eq:CDFgamma}), we have
	\begin{equation}\label{eq:Fout}
		P_{\sf out}=F_{\gamma}(\gamma_{\sf out}).
	\end{equation}
	
	To derive simpler expressions and get more insights into the impact of various system parameters on the behavior, we derive a simple but accurate expression for the outage probability at high values of SNR ($\bar{\gamma}\rightarrow\infty$). In this region, the outage probability can expressed as $P_{\sf out}=(G_{c}\bar{\gamma})^{-G_{d}}$, where $G_{d}$ is the gain in diversity and $G_{c}$ is the gain in coding \cite{Alouinib}. Using \cite[Eq. (8.354.1)]{Grad.}, we get
	\begin{equation}\label{eq:AsympOut1}
		P_{\sf out}^{\infty}\simeq \frac{\sum_{n=0}^{\infty}\frac{(-1)^{n}\left(\sqrt{\frac{\gamma_{\sf out}}{\bar{\gamma}}}\right)^{a+n+1}}{(a+n+1)b^{a+n+1}}}{\Gamma(a+1)}.
	\end{equation}
	As $\bar{\gamma}\rightarrow\infty$, the expression in \eqref{eq:AsympOut1} is only dominated by the first term in summation. Upon considering that, we get
	\begin{equation}\label{eq:AsympOut}
		P_{\sf out}^{\infty}\simeq \left[\frac{b^{2}}{\gamma_{\sf out}[(a+1)!]^{-\frac{2}{(a+1)}}}\bar{\gamma}\right]^{-\frac{(a+1)}{2}}.
	\end{equation}
	
	From \eqref{eq:AsympOut}, it can be noticed that the coding gain is $G_{c}=\frac{b^{2}}{\gamma_{\sf out}[(a+1)!]^{-\frac{2}{(a+1)}}}$ and the diversity order is $G_{d}=\frac{a+1}{2}$. In addition, for Rayleigh fading channels, $G_{d}$ simplifies to $\frac{\pi ^2 N}{32-2 \pi ^2}$. It is worthwhile to mention that, \cite{Yang6} approximated $G_{d}$ of RIS networks as that of MIMO channels $N-1<G_{d}<N$. Here, we show that this only applies when $N$ takes a value equal to 5 or less, but when $N$ is greater than 5, $G_{d}$ could get below the floor $N-1$.

	\subsection{Average Symbol Error Probability}
	Using the CDF of e2e SNR, the average symbol error probability (ASEP) can be obtained as follows \cite{McKay}
	\begin{align}\label{eq:ASEP}
		\mathrm{ASEP}=\frac{p\sqrt{q}}{2\sqrt{\pi}}\int_{0}^{\infty}\frac{\exp\left(-q\gamma\right)}{\sqrt{\gamma}}F_{\gamma}(\gamma) d\gamma,
	\end{align}
	where $p$ and $q$ are constants representing the type of modulation. 
	
	\begin{theorem}
		For RIS-aided network under Nakagami-$m$ fading environment, the ASEP can be achieved as in (\ref{eq:ASEPpq}), shown at the top of this page.
			\begin{figure*}
		\footnotesize
		\begin{equation}\label{eq:ASEPpq}
			\begin{split}
				P_{e}=\frac{p q^{-\frac{a}{2}-1} \left(\frac{1}{b \sqrt{\bar{\gamma}}}\right)^a \left((a+2) b \sqrt{\bar{\gamma}} \sqrt{q} \Gamma \left(\frac{a}{2}+1\right) \, _2F_2\left(\frac{a}{2}+\frac{1}{2},\frac{a}{2}+1;\frac{1}{2},\frac{a}{2}+\frac{3}{2};\frac{1}{4 b^2 \bar{\gamma} q}\right)-(a+1) \Gamma \left(\frac{a+3}{2}\right) \, _2F_2\left(\frac{a}{2}+1,\frac{a}{2}+\frac{3}{2};\frac{3}{2},\frac{a}{2}+2;\frac{1}{4 b^2 \bar{\gamma} q}\right)\right)}{2 \sqrt{\pi } (a+1) (a+2) b^2 \bar{\gamma} \Gamma (a+1)},
			\end{split}
		\end{equation}
		\normalsize
		\hrule height 0.8pt
	\end{figure*}
		
	\end{theorem}

	\begin{proof}
		Upon inserting (\ref{eq:CDFgamma}) in (\ref{eq:ASEP}), evaluating the integral by substitution, then by parts, then by substitution again, and then utilizing the relation \cite[Eq. 06.25.21.0131.01]{Wolfram}, the result given in (\ref{eq:ASEPpq}) is achieved. 
	\end{proof}
	
	\subsection{Channel Capacity}
	Using the PDF of e2e SNR, the channel capacity can be obtained as follows \cite{Vellakudiyan}
	\begin{align}\label{eq:ACC}
		C=(1/\ln(2))\int_{0}^{\infty} \ln(1+\gamma) f_{\gamma}(\gamma) d\gamma.
	\end{align}
	
	\begin{theorem}
		For RIS-aided network under Nakagami-$m$ fading environment, the channel capacity  can be given by (\ref{eq:ACCgamma}).
	\end{theorem}
		\begin{figure*}
		\scriptsize
		\begin{equation}\label{eq:ACCgamma}
			\begin{split}
				C=\frac{1}{\ln (2) \Gamma (a+1)}\Biggl(\frac{\Gamma (a-1) \, _2F_3\left(1,1;2,1-\frac{a}{2},\frac{3}{2}-\frac{a}{2};-\frac{1}{4 b^2 \bar{\gamma}}\right)}{b^2 \bar{\gamma}}
				+\frac{\pi  b^{-a-2} \bar{\gamma}^{-\frac{a}{2}-1} \csc \left(\frac{\pi  a}{2}\right) \, _1F_2\left(\frac{a}{2}+1;\frac{3}{2},\frac{a}{2}+2;-\frac{1}{4 b^2 \bar{\gamma}}\right)}{a+2}\\
				+\frac{\pi  b^{-a-1} \bar{\gamma}^{-\frac{a}{2}-\frac{1}{2}} \sec \left(\frac{\pi  a}{2}\right) \, _1F_2\left(\frac{a}{2}+\frac{1}{2};\frac{1}{2},\frac{a}{2}+\frac{3}{2};-\frac{1}{4 b^2 \bar{\gamma}}\right)}{a+1}
				-2 a^2 \Gamma (a-1) \ln \left(\frac{1}{b \sqrt{\bar{\gamma}}}\right)+2 a \Gamma (a-1) \ln \left(\frac{1}{b \sqrt{\bar{\gamma}}}\right)
				+2 (a-1) a \Gamma (a-1) \psi ^{(0)}(a+1)\Biggr),
			\end{split}
		\end{equation}
		where $\psi^{0}(.)$ gives the $0^{\mathrm{th}}$ polygamma function, which is the $0^{\mathrm{th}}$ derivative of the digamma function.\\
		\normalsize
		\hrule height 0.8pt
	\end{figure*}
	\begin{proof}
		Upon inserting (\ref{eq:PDFgamma}) in (\ref{eq:ACC}), and then utilizing the following representation of the natural logarithm function \cite{Wikipedia}
		\begin{equation}\label{eq:logarithm}
			\ln(1+y)=y_2F_1(1,1;2;-y),
		\end{equation}
		and then using \cite[Eq. 07.23.21.0015.01]{Wolfram}, the result given in (\ref{eq:ACCgamma}) is achieved.
	\end{proof}
	
	\section{Optimum Number of Reflecting Elements}\label{ONRE}
	Here, we derive a closed-form expression for the optimum number of RIS elements $N^{\rm Opt}$ using the derived outage probability in \eqref{eq:Fout}. Since this expression is not convex with respect to number of RIS elements $N$, we assume here that $N^{\rm Opt}$ could be obtained for a predetermined outage performance $P_{\sf out}^{\rm th}$. Hence, the optimization problem can be expressed as 
	\begin{align}
		\min_{N} P_{\sf out} (N) \qquad   \text{s.t.} \ \ P_{\sf out} \leq P_{\sf out}^{\rm th}, \ \ 1\leq N \leq N_{\rm max}. 
	\end{align} 
	Such technique has been adopted in \cite{Larsson2020} for minimum rate. By substituting (3), (4), and (5) in \eqref{opt2}, we obtain
	\begin{align} \label{opt2}
		\min_{N} \frac{\gamma\left(a +1, \frac{\sqrt{\gamma_{\sf out}/ \bar{\gamma}}}{b} \right)}{\Gamma(a+1)} \qquad \text{s.t.} \ \ P_{\sf out} \leq P_{\sf out}^{\rm th}, \ \ 1\leq N \leq N_{\rm max}. 
	\end{align} 
	With the incomplete gamma function upper bound provided in \cite{Jameson2016}, the expression in \eqref{opt2} can be further simplified as
	\begin{align} \label{logapp1}
		\frac{\gamma\left(a +1, \frac{\sqrt{\gamma_{\sf out}/ \bar{\gamma}}}{b} \right)}{\Gamma(a+1)} \leq \frac{\exp(a+1) \left( \frac{\sqrt{\gamma_{\sf out}/ \bar{\gamma}}}{b} \right)^{(a+1)} }{(a+1)^{(a+1)}}, 
	\end{align} 
	and hence \eqref{opt2} can be solved by solving the following 
	\begin{align}
		\frac{\exp(a+1) \left( \frac{\sqrt{\gamma_{\sf out}/ \bar{\gamma}}}{b} \right)^{(a+1)} }{(a+1)^{(a+1)}} = P_{\sf out}^{\rm th},
	\end{align} 
	then, by taking the natural logarithmic of both sides, we get
	\begin{align} \label{logapp}
		&(a+1) + (a+1) \ln \left( \frac{\sqrt{\gamma_{\sf out}/ \bar{\gamma}}}{b} \right) - (a+1) \ln(a+1) = \ln (P_{\sf out}^{\rm th}).
	\end{align}
	This equation can be solved by using any software package such as Matlab. We call this solution as the logarithmic approximation approach as will be discussed later. 
	
	A further simplification to \eqref{logapp} can be done by representing the logarithmic equation by the quadratic polynomial. This can be done by using the curve fitting method where the term $(a+1) \ln(a+1)$ can be approximately expressed by a second degree polynomial as follows
	\begin{align}
		(a+1) \ln(a+1) \approx 0.001248 (a+1)^2 + 5.825 (a+1) -131.4. 
	\end{align}
	
	\begin{figure*}
		\begin{equation}\label{quadapp}
			(a + 1) = \frac{-\left[  4.825 -\ln \left( \frac{\sqrt{\gamma_{\sf out}/ \bar{\gamma}}}{b} \right) \right] \pm \sqrt{\left[  4.825 -\ln \left( \frac{\sqrt{\gamma_{\sf out}/ \bar{\gamma}}}{b} \right) \right]^2 - 4 \times 001248 \times \left[  -131.4 + \ln (P_{\sf out}^{\rm th}) \right] }}{2 \times 0.001248}.
		\end{equation}
		\hrule height 0.8pt
	\end{figure*}

	Then, the problem in \eqref{logapp} can be expressed as
	\begin{align}
		&	(a+1) + (a+1) \ln \left( \frac{\sqrt{\gamma_{\sf out}/ \bar{\gamma}}}{b} \right) - 0.001248 (a+1)^2 - 5.825 (a+1) +131.4 = \ln (P_{\sf out}^{\rm th}),  \\
		&	0.001248 (a+1)^2 + \left[  4.825 -\ln \left( \frac{\sqrt{\gamma_{\sf out}/ \bar{\gamma}}}{b} \right) \right] (a+1)-131.4 + \ln (P_{\sf out}^{\rm th}) = 0.
	\end{align}
	This equation can be solved by using the formula of solving quadratic polynomial equation as shown in \eqref{quadapp} on the top of this page. Hence, the value of $N^{\rm Opt}$ can be evaluated from $a$.

	\section{Simulation and Numerical Results}\label{SNRs}
	Here, we assume $\Omega_{1}=\Omega_{2}=\Omega$ and $m_{1}=m_{2}=m$ in the first 4 figures. We are also utilizing ($p=q=1$) in the ASEP results, which represent BPSK modulation.
	
	Fig. \ref{Pout_SNR_N} validates the derived outage probability expression where obviously a good matching happens between the derived analytical and asymptotic results with simulations, which validates the followed derivation approach. Moreover, we can notice from this figure that the number of meta-surface elements $N$ is clearly impacting the diversity order of the system $G_{d}$, which matches the achieved asymptotic results in Section \ref{EPA}. In terms of saving power, utilizing $N=10$ instead of $N=5$ to achieve an outage probability of $10^{-4}$ can save an amount of almost 11 dB in the SNR. The figure also compares our results with the CLT-based approach where it is obvious that the CLT-based results diverge from the simulation and our results. This is an indication on the inaccuracy of the CLT-based approach in deriving channel distribution of RIS-assisted networks. 
	
	The impact of fading parameter $m$ on the ASEP performance is studied in Fig. \ref{BER}. Obviously, $m$ is affecting the diversity gain $G_{d}$, where as $m$ increases, $G_{d}$ increases resulting in a better behavior. This is expected, as $m$ increases, the quality of fading channel improves, and hence, better the achieved behavior.
	
	Fig. \ref{DO} portrays the diversity order of the system versus $N$ for different values of $m$. We can see from this figure that $m$ has higher impact on $G_{d}$, and hence, the system performance than the number of reflecting elements $N$. For example, a diversity order of $G_{d}=8$ is achieved when $N=10$ and $m=1$, whereas a diversity order of $G_{d}=9.75$ is almost achieved when $N=1$ and $m=10$. This clearly shows that $G_{d}$ could go below than $mN-1$ by a big amount when $N$ is larger than $m$, especially at high values of $N$ and $m$, but is almost floored by $mN-1$ when $m$ is larger than $N$ even at high values of $N$ and $m$.
	
	Fig. \ref{C1} studies the effect of $\Omega$ on the behavior. Obviously, as $\Omega$ increases, better the achieved capacity. In addition, we can see that as $\Omega$ continues increasing, the gain achieved in system capacity gets smaller.
	
	In the next two figures, we show how some key performance factors impact the determination of the value of $N^{\rm Opt}$. Here, we use the following parameters: operating frequency is 2 GHz, distances between the source and RIS, and between the RIS and destination are set to be 5 m and 10 m, respectively, and the channel shape and scale parameters of first and second links are, respectively set as $m_1 = 2$, $m_2 = 1$, $\Omega_{1}=0.4123$, and $\Omega_{2}=0.8973$. Fig. \ref{fig6} studies the effect of $\gamma_{\sf out}$ in determining $N^{\rm Opt}$ for SNR $=$ 15 dB and $P_{\sf out}^{\rm th} = 10^{-20}$. It is clear that increasing $\gamma_{\sf out}$ increases $N^{\rm Opt}$ as more elements are needed to be involved to achieve the required outage performance. The results also show the good matching between the exact and the logarithmic approximation approaches in terms of  $N^{\rm Opt}$ over a wide range of $\gamma_{\sf out}$. Clearly, the quadratic polynomial approximation starts with a high percentage error about 85.71\% at low $\gamma_{\sf out}$, however as $\gamma_{\sf out}$ keeps increasing this error decreases until it reaches almost 18.75\% at high $\gamma_{\sf out}$ values. This deficiency in the quadratic polynomial approach comes as an expense to its simplicity in finding the optimum number of elements.  
	
	Fig. \ref{fig7} studies the impact of $P_{\sf out}^{\rm th}$ on obtaining  $N^{\rm Opt}$ for the three proposed optimization approaches for SNR $=$ 15 dB and $\gamma_{\sf out} = 0$ dB. It is clear from this figure that increasing $P_{\sf out}^{\rm th}$ decreases $N^{\rm Opt}$ as the required performance is more degraded. Also, this figure illustrates that it is very important to choose the suitable $P_{\sf out}^{\rm th}$ to achieve a good matching in the results of the three optimization approaches in terms of $N^{\rm Opt}$. Finally, as $P_{\sf out}^{\rm th}$ decreases, the accuracy of the logarithmic and the quadratic polynomial approximation approaches in finding the optimal number of meta-surface elements increases.

	
	%
	\section{Conclusion}\label{C}
	This letter investigated the outage and error rate performances, in addition to channel capacity of RIS-assisted networks in Nakagami-$m$ environment. The paper also provided closed-form expression for the optimum number of reflecting elements. In addition, it studied the asymptotic behavior in terms of diversity and coding gains. Results showed that the system can provide a diversity order of $\frac{a+1}{2}$, where $a$ is a function of the Nakagami-$m$ fading parameter $m$ and the number of meta-surface elements $N$. Moreover, findings illustrated that $m$ is more influential on the diversity gain and behavior than $N$.
	
	\clearpage
	
\begin{figure}[htb!]
		\centering
		\includegraphics[scale=0.7]{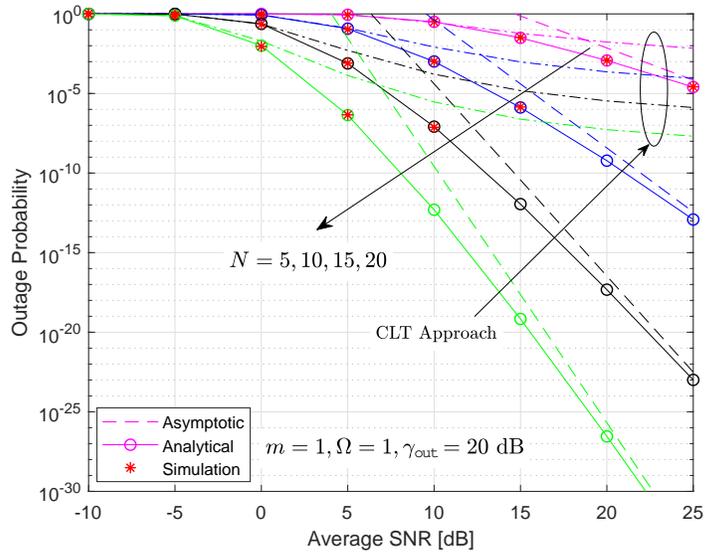}
		\caption{$P_{\sf out}$ vs SNR for various
			$N$.}\label{Pout_SNR_N}
	\end{figure}
	
	\begin{figure}[htb!]
		\centering
		\includegraphics[scale=0.7]{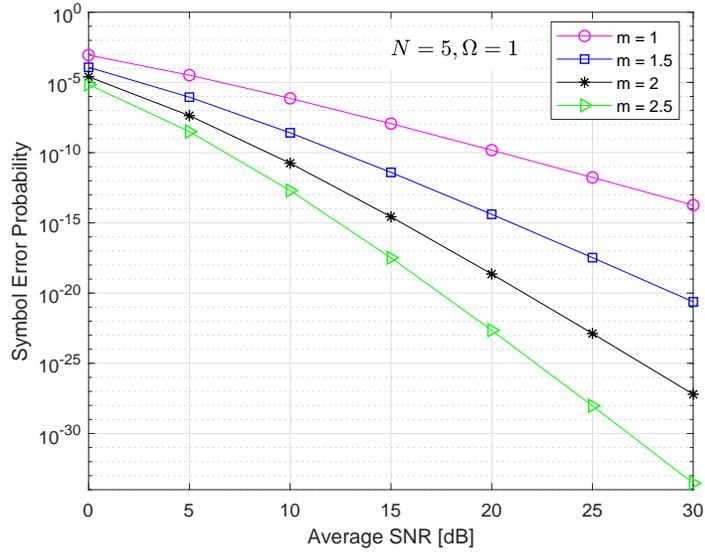}
		\caption{ASEP vs SNR for various
			$m$.}\label{BER}
	\end{figure}
	
	\begin{figure}[htb!]
		\centering
		\includegraphics[scale=0.7]{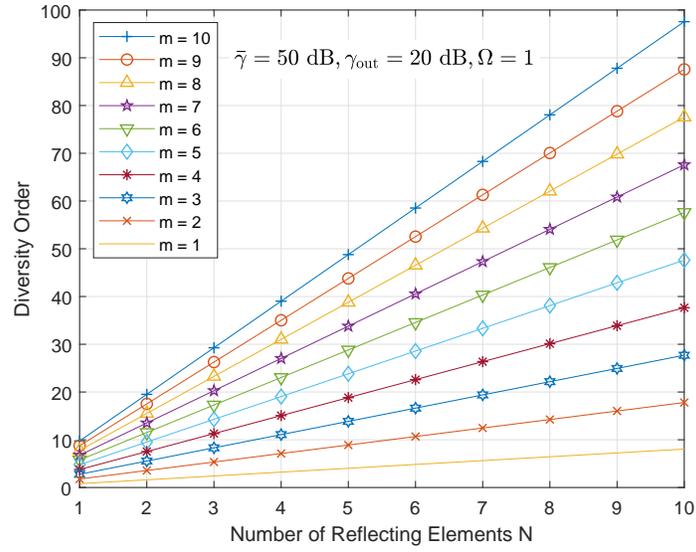}
		\caption{$G_{d}$ vs $N$ for various
			$m$.}\label{DO}
	\end{figure}
	
	\begin{figure}[htb!]
		\centering
		\includegraphics[scale=0.7]{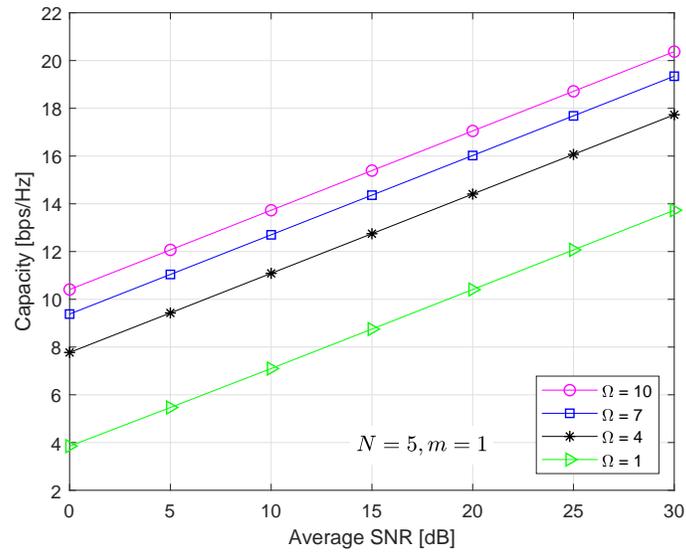}
		\caption{$C$ vs SNR for various $\Omega$.}\label{C1}
	\end{figure}
	
	\begin{figure}[!htb]
		\centering
		\includegraphics[scale=0.7]{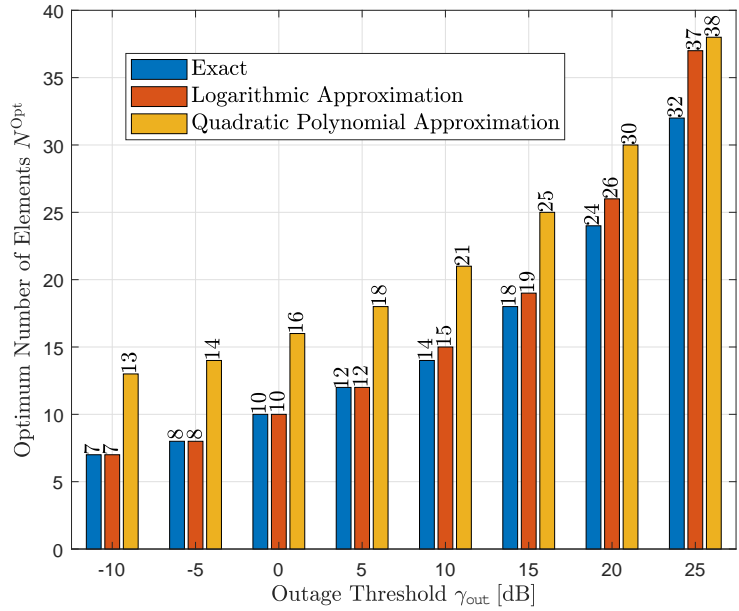}
		\caption{$N^{\rm Opt}$ vs $\gamma_{\sf out}$ for different optimization approaches.}
		\label{fig6}
	\end{figure}
	
	\begin{figure}[!htb]
		\centering
		\includegraphics[scale=0.7]{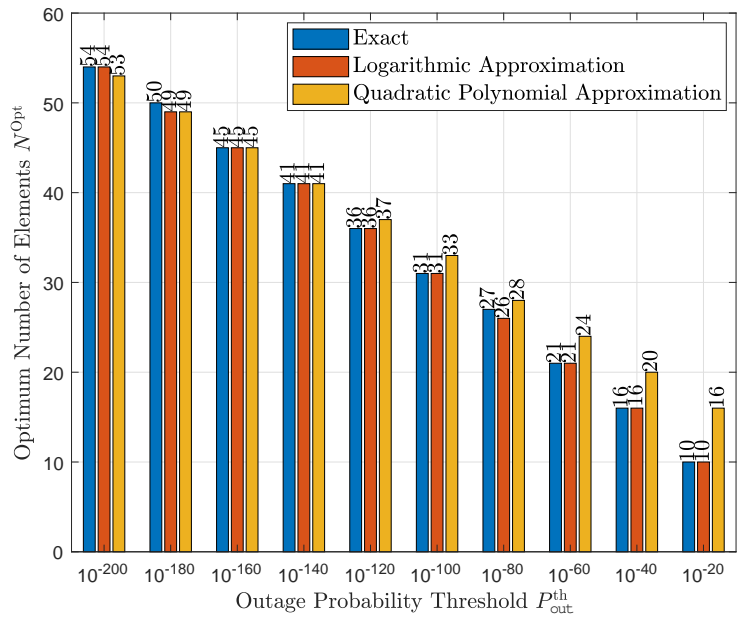}
		\caption{$N^{\rm Opt}$ vs $P_{\sf out}^{\rm th}$ for different optimization approaches.}
		\label{fig7}
	\end{figure}

	
	
	


	\ifCLASSOPTIONcaptionsoff
	\fi

	
	
	%
	

	\clearpage
\end{document}